\documentclass[a4paper]{article}
\usepackage{hyperref}
\usepackage{bbding}
\usepackage[T1]{fontenc}
\usepackage[utf8]{inputenc}
\usepackage{aurical}
\usepackage{huncial}
\usepackage{linearb}
\usepackage{textcomp}
\usepackage{amsmath}
\usepackage{amssymb}
\usepackage[amsmath,amsthm,thmmarks]{ntheorem}
\usepackage{amsfonts}
\usepackage{amssymb,epic,eepic}
\usepackage{stmaryrd}
\usepackage{capt-of}
\usepackage{graphicx}
\usepackage[usenames,dvipsnames]{pstricks}
\usepackage{epsfig}
\usepackage{pst-grad}
 \usepackage{pst-plot}
\usepackage{here}
\usepackage[justification=centering]{caption}

\newtheorem{theorem}{Theorem}[section]

\newtheorem{corollary}[theorem]{Corollary}
\newtheorem{definition}[theorem]{Definition}
\newtheorem{remark}[theorem]{Remark}

\newtheorem{lem}{Lemma}[section]

\theoremstyle{definition}

\newtheorem{ex}[lem]{Example}
\theoremstyle{remark}

\usepackage{mathtools}

\begin{document}
\title{Mosaics of Combinatorial Designs for 
Semantic Security on Quantum Wiretap Channels }
\author{Holger Boche\\
Lehrstuhl f\"ur Theoretische
Informationstechnik,\\
 Technische Universit\"at M\"unchen,\\
Munich, Germany\\
Excellence Cluster Cyber Security\\
 Cyber 
 Security in the Age of Large-Scale Adversaries,\\
 Ruhr-Universit\"at Bochum,\\
  Bochum, Germany\\
 Munich Center for Quantum Science and Technology (MCQST),\\ 
 Munich, Germany\\
boche@tum.de \and Minglai Cai\\
Quantum Information Group\\
Universitat Aut\`{o}noma
 de Barcelona,\\ 
 Barcelona, Spain \and 
Moritz Wiese \\
Lehrstuhl f\"ur Theoretische
Informationstechnik,\\
 Technische Universit\"at M\"unchen,\\
Munich, Germany\\
Excellence Cluster Cyber Security\\
 Cyber 
 Security in the Age of Large-Scale Adversaries,\\
 Ruhr-Universit\"at Bochum,\\
  Bochum, Germany\\
wiese@tum.de
}
\maketitle
\begin{abstract}
We study  semantic security
for classical-quantum channels.
Our security functions are
functional forms
of mosaics of combinatorial designs.
We extend methods  of \cite{Bo/Wi}
for classical channels to classical-quantum channels
to demonstrate that mosaics of designs ensure semantic security
for classical-quantum channels, and are also capacity achieving
coding scheme..
The legitimate channel users share an additional public resource,   more precisely,
   a seed  chosen uniformly at random.
  An advantage of these modular wiretap codes 
is  that we   provide explicit code constructions 
that can be implemented
in practice for every channels, giving
an arbitrary public code.
\end{abstract}

\tableofcontents
\section{Introduction}

We investigate the transmission of messages from a sender to a 
legitimate receiver to ensure semantic security.
The model of a wiretap channel adds a third party to the communication 
problem with focus on secure communication, meaning 
communication without that third party getting to know the messages. This 
model was first introduced by Wyner in \cite{Wyn}.
A classical-quantum channel with an eavesdropper is called a classical-quantum wiretap channel, 
and its secrecy capacity has been determined in \cite{Ca/Wi/Ye} and \cite{De}.

In most of the previous works only strong security is  required, 
meaning that given a  uniformly distributed message sent through the channel, 
the eavesdropper shall obtain no information about it. 
This is the more common secrecy criterion in the quantum information theory.  
However, our goal is 
a stronger security property formalized by \cite{Be/Te/Va}, namely the  semantic security.
     Semantic security in the information theory imposes the eavesdropper 
to gain no information for any distribution of the messages, not just the uniform one.
In cryptography, semantic security is a security goal for key cryptosystem 
such that the eavesdropper cannot distinguish the given encryption  of any two messages. 
(cf. \cite{Go/Mi} and \cite{Be/Te/Va}).
The equivalence between semantic security
and message indistinguishability under chosen-plaintext attack has been shown in \cite{Go/Mi2}.\vspace{0.2cm}

Most of these pre-works 
 merely delivered an existence proof that
there exists secure codes achieving the security capacity
formula, but do not answer the question how
these codes can be construct.
However, in recent years, explicit  code constructions 
become more important in secure network design
and quantum communications 
(cf.  \cite{Re/Du/Re} for an example).
 An essential aspect of code
construction  is that  it can be implemented
in practice.
It is expected, that information theoretical
 security will play an important roll in
 future communication systems. It is the very
 technique to achieve security by design, which is
 already a key requirement for 6G (\cite{Fe/Bo1} and \cite{Fe/Bo2}).
 Furthermore, it is expected, that quantum communication and the use
 of quantum resources will be important for achieving 
 the design goals of 6G (\cite{Fi/Bo}).
Therefore, the construction of secure codes for
quantum wiretap channels is an important  requirement.

\cite{Be/Te/Va} investigated 
semantic security for certain special 
classical wiretap channels and could provide
 semantic secure code constructions for these channels.
     \cite{Li/Ya/Li} provided  semantic secure polar code constructions 
for Gaussian wiretap channels.
Please also see \cite{Bl/Ha/Th} for further explicit semantic 
secure code constructions of certain  channels, and \cite{Be/Li/Ja/Jo} for 
another approach of secure  code
design of Gaussian wiretap channels.\vspace{0.2cm}

In this work,
we consider modular wiretap codes constructed
from an  arbitrary transmission code and
a security function. This modular code concept
is  not only  an existence statement, 
but actually show explicitly how  to construct semantic
secrecy capacity achieving
codes for \bf every \rm classical-quantum wiretap channel, as long as
\bf any \rm reliable public  transmission code
is given.

  Using modular code with hash functions, defined via the function inverses in terms of
group homomorphisms, as security function, 
    \cite{Hay/Mat}
 showed   semantic security 
      for  classical channels
(cf. Example \ref{ex2}).
    The  technique of \cite{Hay/Mat} has been applied  in \cite{Ha15}
 for  additive  fully quantum channels,
when the eavesdropper has access to the whole environment.
     However, it is unknown whether the seed required in \cite{Hay/Mat}
 can be as short as for the security functions in this work.
    Furthermore, the results in \cite{Ha15}
are limited to linear codes and to
additive channels. 

\begin{center}\begin{figure}[H]
\includegraphics[width=0.9\linewidth]{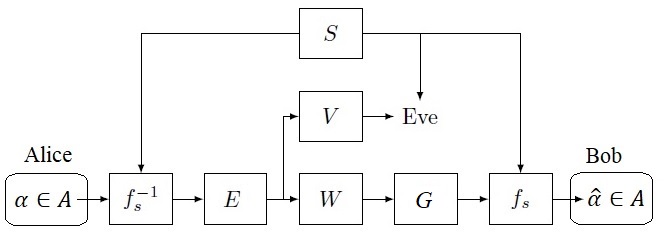}
\caption{
 The classical-quantum wiretap channel scenario.
   $W$ denotes the legal
channel, while $V$ is the wiretap.
     $(E,G)$ is
 a reliable code. $f$ is 
a functional form of mosaics of combinatorial design.
 The  seed $s$ has to be known to the
sender and receiver, and may be known by the eavesdropper.
}\label{avcqwcwabsr}
\end{figure}\end{center}

   A modular code
for the classical-quantum wiretap channel is illustrated in the Figure \ref{avcqwcwabsr}.
    A reliable public  code $\mathcal{C}_{public}$
$=(E,\{G_{\alpha}:\alpha\})$
from the sender to the intended
receiver with input alphabet ${\mathcal X}$,
consisting of a encoder $E$ and
a set of decoder operators $\{G_{\alpha}:\alpha\}$,
 is given. $f$ is a hash function.
     The sender and the intended receiver 
have to share a seed $s$, chosen uniformly at random.   
Given any message $\alpha$ and seed $s$, the 
sender randomly chooses a preimage $x$, satisfying $f(s,x)=\alpha$,
and send through the channel via the given encoder. 
We assume  that the receive can decoder $x$ with
decoding error $P_e(\mathcal{C}_{public})$.
Since  $s$ is known by
the receiver, $a$  can recovered with
decoding error $P_e(\mathcal{C}_{public})$. 
We emphasize that the seed is not a secret key, 
since we do not require it to be unknown to the eavesdropper.
 There is a   separation of
the security task and the reliability task:
Since the intended receiver knows $s$, he can recover $\alpha$
with errors $P_e(\mathcal{C}_{public})$.
The reliability task depends here only
on  $\mathcal{C}_{public}$, but not on $f$.
    On the other hand,
the security task depends here only
on  $f$,  but not on $\mathcal{C}_{public}$.
This is a notable advantage of  
this modular code, namely since
the reliability task depends here only
on  $\mathcal{C}_{public}$, the efficiency 
of  reliability  of the  modular code is the same
as the reliability 
of the given public  code. 
The existence of efficient reliable public  codes
and their constructions have been already extensively
analyzed in the above cited pre-works. 
On the other hand, since 
the security task depends here only
on  $f$, a functional
form of mosaic of combinatorial design, the efficiency 
of  security   of the  modular code can be
analyzed independent of the given public  code,
when we can show that the functional
forms of mosaic  designs  guarantee semantic security.
We will show that the functional
forms of mosaic  designs always proves semantic security for
modular code made of any  public  code
 in this work. Together with the well-know results of efficient reliable public  codes,
 we show a most general  semantic secure code constructions
 that can be implemented
in practice.

Using modular code with  
functional
forms of mosaics of combinatorial designs
as hash functions,  \cite{Bo/Wi}
showed semantic security for classical channels.    
 In this paper we     extend this method
to classical-quantum channels,
   i.e., we construct a modular code, where the hash functions,
  which are
functional
forms of mosaics of combinatorial design,
are used as  security function.
   A functional form of mosaics of combinatorial design
has the  form $f:\mathcal S\times\mathcal X\to\mathcal A$. 
Here, $\mathcal S$ is a seed set, and $\mathcal A$ is the set of messages. 
Every
pre-image  is
 the incidence relation of a balanced incomplete block design (BIBD) or a
group divisible design (GDD). Such a function
defines a mosaic of designs  $\{D_{\alpha}: \alpha\in A\}$, 
which is a family of designs on  $\mathcal{X}$
and $\mathcal{S}$ satisfying that every pair $(x,s)$  is incident in a
unique  $D_{\alpha}$.
We bound the information leak  to the
eavesdropper and show  semantic security.
The approach 
 can be applied on    \bf  any \rm  reliable  transmission code.

   We apply the standard
derandomization technique.
The original derandomization technique
works this way (cf.  the code concept for  arbitrarily varying classical-quantum channels in
\cite{Ahl/Bli}):
For every message, a two-part code word,  which consists of a non-secure code word
 and a    modular code word, is built. The first part
is used to
to produce  the seed.
The second  part
 is used to transmit the message. However, 
this standard  technique
may cause significant rate loss
when size of the needed seed set is too large.
Follow the idea
of \cite{Wi/Bo} and \cite{Bo/Ca/De/Fe/Wi},
we reduce the total size of channel uses by
reusing the seed.  
Instead of one  two-part code word for
every single message, we build for multiple
messages a block of code words. The 
leading  public code word is used to
create one single seed.
The following parts are
multiple  semantic secure modular code words,
which share this one single seed.
With this code concept, the  semantic secrecy capacity
for classical-quantum channel
given in \cite{Bo/Ca/De/Fe/Wi} can be achieved.

As mentioned above, a  notable advantage of the modular wiretap codes presented
in this paper is that they
separate the the task of reliable message transmission 
from the sender to the intended receiver from that of securing 
the message from the eavesdropper. Moreover, the application of 
the functional form of a mosaic of designs does not require any
quantum operations. Hence  semantic security, 
can be realized in classical software or hardware. 
In particular, security functions do not have to appear in the lowest, physical layer of the 
open systems interconnection (OSI) networking model (cf. \cite{Ta/We}). 
Instead, they could be implemented on the medium access control (MAC) 
layer, for example. As soon as the data which are to be transmitted securely over the data link are known, the randomized inverse of the security function can be applied. The physical layer can be left untouched. This allows for an easy integration of information theoretic security into existing systems, especially if the higher layers are realized in software. In this case, a simple software update is sufficient in order to enable the support of information theoretic security.


\section{Basic Notations and Definitions}
\label{BNaD}

\subsection{Communication Scenarios}
For a finite set $B$, we denote the
set of probability distributions on $B$ by $P(B)$.
Let $\rho_1$ and  $\rho_2$ be  Hermitian   operators on a  finite-dimensional
complex Hilbert  space ${\cal H}$.
We say $\rho_1\geq\rho_2$ and $\rho_2\leq\rho_1$ if $\rho_1-\rho_2$
is positive-semidefinite.
 For a finite-dimensional
complex Hilbert space  ${\cal H}$, we denote
the
set of  density operators on $G$ by
\[\mathcal{S}(G):= \{\rho \in \mathcal{L}(G) :\rho  \text{ is Hermitian, } \rho \geq 0_{{\cal H}} \text{ , }  \mathrm{tr}(\rho) = 1 \}\text{ ,}\]
where $\mathcal{L}({\cal H})$ is the set  of linear  operators on ${\cal H}$, and $0_{{\cal H}}$ is the null
matrix on ${\cal H}$. Note that any operator in $\mathcal{S}({\cal H})$ is bounded.

 Throughout the paper the logarithm base   is  $2$.\vspace{0.2cm}

Let $H$ be a finite-dimensional complex Hilbert space.  
A \textbf{classical-quantum channel}    is
a  map $W: \mathcal{X}\rightarrow\mathcal{S}({\cal H})$,
$ x \rightarrow W(x)$.

\begin{definition}
Let $\mathcal{X}$ be a finite set. 
Let  $H$ and $H'$ be finite-dimensional complex Hilbert spaces. 
Let $W$ be a classical-quantum channel $\mathcal{X} \rightarrow \mathcal{S}(H)$ and 
${V}$ be a classical-quantum channel $\mathcal{X} \rightarrow \mathcal{S}(H')$. 
We call  the  classical-quantum channel pair  $(W,{V})$ 
a \textbf{classical-quantum wiretap channel}. 
The legitimate receiver accesses the output of the first channel $W$, 
and the eavesdropper observes the output of the second channel ${V}$ in the pair, 
respectively.
\end{definition}

For a quantum state $\rho\in \mathcal{S}({\cal H})$ we denote the von Neumann
entropy of $\rho$ by $S(\rho)=- \mathrm{tr}(\rho\log\rho)$. 
Let $\Phi := \{\rho_x : x\in \mathcal{X}\}$  be a set of quantum  states
labeled by elements of $\mathcal{X}$. For a probability distribution  $Q$
on $\mathcal{X}$, the Holevo $\chi$ quantity, or Holevo information, is defined as
\[\chi(Q;\Phi) := S\left(\sum_{x\in \mathcal{X}} Q(x)\rho_x\right)
-\sum_{x\in \mathcal{X}} Q(x)S\left(\rho_x\right)\text{ .}\]

Let $\rho$ and $\sigma$ be two positive semi-definite operators.
The quantum relative entropy between  $\rho$ and $\sigma$ is defined as follows:
\[D(\rho\parallel\sigma) 
:= \mathrm{tr}\rho\left(\log \rho- \log \sigma \right)\]
if $supp(\rho) \subset supp(\sigma)$, and $=\infty$ otherwise.

The R\'enyi $2$-relative entropy 
between  $\rho$ and $\sigma$ is defined as
\[D_{2}(\rho\parallel\sigma) := \log \mathrm{tr}\left(\rho^{2}\sigma^{-1}\right)\]
if $supp(\rho) \subset supp(\sigma)$, and $=\infty$ otherwise.

It is well-known that for any density operators $\rho$ and $\sigma$, 
 it holds (cf. \cite{Mo/Da})
\begin{equation}D_{\alpha}(\rho\parallel\sigma)
\leq D_{\alpha'}(\rho\parallel\sigma)\text{ .}\label{ldarsld}\end{equation}
Furthermore it holds
\begin{equation}\lim_{\alpha\nearrow 1} D_{\alpha}(\rho\parallel\sigma) = 
\lim_{\alpha\searrow 1} D_{\alpha}(\rho\parallel\sigma) =
D(\rho\parallel\sigma)\text{ .}\label{ldarsld2}\end{equation}

\begin{definition}
Let $\mathcal{A}_n=\left\{1,\cdots,J_n\right\}$.
 An $(n, J_n)$  \textbf{code} $\mathcal{C}$ for  $(W,{V})$
consists of a stochastic encoder $E$ : $\mathcal{A}_n \rightarrow P({\mathcal{X}}^n)$, 
${\alpha}\rightarrow E(\cdot|{\alpha})$,
 specified by
a matrix of conditional probabilities $E(\cdot|\cdot)$, and
 a positive operator-valued measure (POVM) $\left\{G_{\alpha}: \alpha\in \mathcal{A}_n\right\}$
on ${H}^{\otimes n}$ which we call the decoder operators.

The  average probability of the decoding error of a code $\mathcal{C}$ is defined as 
\[P_e(\mathcal{C}, n) := 
 1- \frac{1}{J_n} \sum_{\alpha\in \mathcal{A}_n}
E(x^n|{\alpha})\mathrm{tr}(W^{\otimes n}(x^n)G_{\alpha})\text{ .}\]

The  maximal probability of the decoding error of a code $\mathcal{C}$ is defined as 
\[P_e^{m}(\mathcal{C}, n) := 
 1-  \max_{\alpha\in \mathcal{A}_n}
E(x^n|\alpha)\mathrm{tr}(W^{\otimes n}(x^n)G_{\alpha})\text{ .}\]

For any random variable $A$ on the messages set $\mathcal{A}_n$ ,
the  \textbf{leakage}   of $\mathcal{C}$ with respect to $A$ is defined as
$\chi(A;Z)$, here $Z$  $=\{Z(\alpha): \alpha\in \mathcal{A}_n\}$ are the resulting  quantum states
at the output of wiretap channels  $V$. 
\end{definition}

A code is created by the sender and the legal receiver before the
message transmission starts. The sender uses the encoder to encode the
message that he wants to send, while the legal receiver uses the
decoder operators on the channel output to decode the message.

\begin{definition}
A non-negative number $R$ is an achievable \textbf{semantic secrecy
rate}  for the classical-quantum wiretap channel
 $(W,{V})$ under the average (or the maximum) error 
criteria if for every $\epsilon>0$, $\delta>0$,
$\zeta>0$ and sufficiently large $n$ there exists an  $(n, J_n)$
code $\mathcal{C} = (E, \{G_{\alpha}: \alpha\in \mathcal{A}_n\})$  such that $\frac{\log
J_n}{n}
> R-\delta$, and     every     variable $Q$ with 
arbitrary distribution  on  $\mathcal{A}_n$
such that
\[P_e(\mathcal{C},n)< \epsilon\text{ .,}\]
 \[\text{ (or, }P_e^{m}(\mathcal{C},n)< \epsilon\text{ ,)}\] and
\[\chi\left(Q;Z^n\right) < \zeta\text{ ,}\] respectively.

The supremum over all achievable semantic secrecy rates
under the average and maximum error criteria  of
$(W,{V})$ is called the semantic secrecy
capacity of $(W,{V})$, denoted by
$C_{sem}((W,{V}))$ and $C_{sem}^{m}((W,{V}))$ , respectively.
\end{definition}
\vspace{0.3cm}

\subsection{Mosaics of combinatorial designs}
\label{Mocd}
We define the  mosaics of combinatorial designs
in the same way as in \cite{Bo/Wi}. For the sake
of completeness, the definitions in \cite{Bo/Wi}
are stated below.\vspace{0.2cm}

Let $\mathcal{X}$ and $\mathcal{S}$ be finite sets. 
An incidence structure $D = (\mathcal{X}, \mathcal{S}, I)$ 
on $(\mathcal{X}, \mathcal{S})$ is
determined by the incidence relation $I$ 
on $\mathcal{X} \times\mathcal{S}$. 
An incidence structure $(\mathcal{X}, \mathcal{S}, I)$
is called empty if $I = \emptyset$. 
If  $xIs$, 
then $x$ and $s$ are called incident. The incidence
matrix of an incidence structure $D = (\mathcal{X}, \mathcal{S}, I)$ 
is the $01$-matrix $N$ with rows
indexed by $\mathcal{X}$ and columns indexed by $\mathcal{S}$ such that 
$N_{x,s} = 1$ if and only if $x$
and $s$ are incident in $D$.

A mosaic of incidence structures on $(\mathcal{X}, \mathcal{S})$ 
is a family $M = \{D_{\alpha}: \alpha\in \mathcal{A}\}$ of
nonempty incidence structures on  $(\mathcal{X}, \mathcal{S})$ 
 such that for every pair $(x,s)$ there exists
a unique incidence structure $D_{\alpha}$ in which $x$
 and $s$ are incident. We call $\mathcal{A}$ the color
set of $M$. Every $D_{\alpha}$ is called a member of $M$. 
If $N_{\alpha}$ is the incidence matrix of $D_{\alpha}$,
then $\sum_{\alpha} D_{\alpha}= J$, here  $J$ is the all-ones matrix
of appropriate dimensions.

Any function $f : \mathcal{X}\times \mathcal{S} \rightarrow \mathcal{A}$ induces a mosaic 
$\{D_{\alpha}: \alpha\in \mathcal{A}\}$  of incidence structures,
where $x$ and $s$ are incident in  $D_{\alpha}$ if and only if 
$f(x,s) = \alpha$. We say that $f$ is the
functional form of this mosaic. Clearly, every mosaic 
$\{D_{\alpha}: \alpha\in \mathcal{A}\}$  on $(\mathcal{X}, \mathcal{S})$ has a
functional form $f : \mathcal{X}\times \mathcal{S} \rightarrow A$.

We consider the case where every $D_{\alpha}$ 
is a combinatorial design. In the context
of designs, we will call $\mathcal{X}$  
the point set and $\mathcal{S}$  the block index set. 
We set $v = |\mathcal{X}|$. A $(v, k, r)$ tactical configuration 
on $(\mathcal{X}, \mathcal{S})$ is an incidence structure
where every point $x$ is incident with precisely $r$ 
block indices and every block index
$s$ is incident with precisely $k$ points. It holds that
\begin{equation} |\mathcal{S}|k = vr\text{ .}\label{skvr}\end{equation} \vspace{0.3cm}

A $(v, k, \lambda)$ balanced incomplete block design (BIBD) is an incidence structure
on $(\mathcal{X}, \mathcal{S})$ such that every $s\in\mathcal{S}$ is incident with precisely $k$ points 
from $\mathcal{X}$, and such that any
two distinct points from $\mathcal{X}$ are incident with precisely $\lambda$ 
common block indices. Every
$(v, k, \lambda)$ BIBD is a $(v, k, r)$ tactical configuration, where
\[r(k - 1) = \lambda(v - 1) \text{ .}\]
The key equality when we want to establish security using a security function which is
the functional form of a mosaic of BIBDs is that the incidence matrix $N$ of a $(v, k, \lambda)$ 
BIBD satisfies
\begin{equation} NN^*  = (r - \lambda)id + 
 \lambda J\text{ ,}\label{nnr-li2}\end{equation} 
here $id$ is the identity matrix of appropriate dimensions.\vspace{0.2cm}

A $(u, m, k, \lambda_1, \lambda_2)$ group divisible design (GDD) 
is based on a partition of $\mathcal{X}$ 
into $m$ point classes of size $u$
each, so $v = um$. Every block index is incident 
with precisely $k$ points, and two
points are incident with $\lambda_1$ common block 
indices if they are contained in the
same point class and with $\lambda_2$ block indices 
otherwise. A $(u, m, k, \lambda_1, \lambda_2)$ GDD
is a $(v, k, r)$ tactical configuration for $r$ satisfying
\[r(k - 1) = \lambda_1(u - 1) + \lambda_2(m - 1)u \text{ .}\]

Let $C$ be
the $01$-matrix with rows and columns indexed by $\mathcal{X}$ 
which has a $1$ in the $(x, x')$
entry if and only if $x$, and $x'$ are contained in the same point 
class. With a suitable
ordering of the elements of $\mathcal{X}$, 
this is a block diagonal matrix with $m$ all-ones
matrices of size $u$ each on the diagonal. Then
\begin{equation} NN^*  = (r - \lambda_1)id + 
(\lambda_1 - \lambda_2)C + \lambda_2J\text{ .}\label{nnr-li}\end{equation} \vspace{0.3cm}

\subsection{Modular Codes}\label{MC}
 
 \begin{definition}\label{abnjaqcftavc}
 Let $S$ be a uniformly distributed random variable on
 $\mathcal{S}$.
A common randomness code is a set
of $|\mathcal{S}|$  codes
 $\left\{\mathcal{C}^s = 
\bigl(E^s, \left\{G_{\alpha}^s,: \alpha\in \mathcal{A}_n\right\}\bigr):s\in \mathcal{S}\right\}$, 
labeled
by  $s$, the common randomness.
 
\end{definition}

\begin{definition}A non-negative number $R$ is an achievable secrecy rate  for the
classical-quantum wiretap channel
$(W,{V})$ under common randomness   quantum coding if  for
every
 $\delta>0$, $\zeta>0$, and  $\epsilon>0$, if $n$ is sufficiently large,
there is an $(n, J_n)$  common randomness 
 code $(\{\mathcal{C}^{s}:s\in \mathcal{S}\})$  such that
$\frac{\log J_n}{n} > R-\delta$, and
\[  \frac{1}{\left|\mathcal{S}\right|} \sum_{s\in \mathcal{S}}P_{e}(\mathcal{C}^{s},n) < \epsilon ~~~
(\text{ or }\frac{1}{\left|\mathcal{S}\right|} 
\sum_{s\in \mathcal{S}}P_{e}^{m}(\mathcal{C}^{s},n)\text{ , respectively,})\]
\[\max_A\frac{1}{\left|\mathcal{S}\right|} \sum_{s\in \mathcal{S}}
\chi\left(A,Z_{\mathcal{C}^{s}}\right)< \epsilon\text{ .}\] This
means that we do not require the common randomness  to be secure
against eavesdropping.\end{definition}\vspace{0.2cm}

We define a modular code
as follows.

\begin{definition}\label{defmod}
Let $\bigl(E, \left\{G_{x}: x\in \mathcal{X}_n\right\}\bigr)$ 
be a $(n, |\mathcal{X}_n|)$  code.
Let $f$ be a function
  $\mathcal{S}\times\mathcal{X}_n\to\mathcal{A}_n$.
We define the  modular code $\left\{\mathcal{C}^s = 
\bigl(E^s, \left\{G_{\alpha}^s,: \alpha\right\}\bigr):s\right\}$
to be the common randomness code 
such that for every $s$ and $\alpha$ we have:\\
$\bullet$ $E^s(x^n|\alpha)$ is the uniform distribution over
$\{x: f(s,x)=\alpha\}$\\
$\bullet$ $G_{\alpha}^s$ 
$=\sum_{f(s,x)=\alpha}G_{x}$.\\
   We call   $f$ the security function.  
\end{definition}

\section{Main Results}
\label{MR}

Assume a reliable $(n, |\mathcal{M}_n|)$ transmission code $\mathcal{C}_{public}$
with input 
alphabet  ${\mathcal X}$ is given.
The security function for the modular code is an onto
  $f:\mathcal S\times\mathcal X\to\mathcal A$ fictional form
of mosaic of combinatorial design. Here ${\mathcal A}$ 
is the set of our confidential messages.\vspace{0.2cm}

        As mentioned above, 
there is a   separation of
the security task and the reliability task:
Since the intended receiver knows the seed, he can recover the message
with error $P_e(\mathcal{C}_{public}, n)$.
The reliability task depends here only
on  $\mathcal{C}_{public}$, but not on $f$.
The security task, on the other side, only 
depends on $f$. We may analyze the secrecy task
 independent of   $\mathcal{C}_{public}$.
Thus, we consider the security for
the message transmission
``$\alpha$ $\rightarrow V \circ E\circ f^{-1}(\alpha)$''
(cf. Figure \ref{avcqwcwabsr}) for any encoder  $E$
instead of considering only the 
eavesdropper's channel $V$,
 i.e., we assume that the encoder  $E$ of the given
transmission code becoming part 
of eavesdropper's channel.
By this way, we can consider
the security under the assumption
that the legal receiver is already able
to decode the massage.
       By Definition \ref{defmod}, for any modular code,
the sender  has to choose a 
$x$ satisfying $\{x: f(s,x)=\alpha\}$.
Thus, when $s$ is the outcome of $S$,
for every encoder $E$, the 
resulting quantum states at the
output of eavesdropper's message transmission 
``$V \circ E\circ f^{-1}$'' 
are $\{Z_{s}(\alpha):\alpha\in \mathcal{A}\}$, where
$Z_{s}(\alpha) := \frac{1}{k}\sum_{x: f(s,x)=\alpha}V(x)$.
  We will show in Section \ref{SSbMoD} that
using security function 
of mosaic of combinatorial design ensures
semantic security, by delivering a security bound 
in terms of the Holevo quantity to the eavesdropper's
resulting states at the outcome of $V$ in Theorem \ref{lfsoma}.\vspace{0.2cm}

It is clear that the semantic
secrecy rate is not the same as the rate
of $\mathcal{C}_{public}$. We will consider
the secrecy rate in Section \ref{Groaom}.

\subsection{Semantic Security by Mosaics of Designs}\label{SSbMoD}

\vspace{0.2cm}

\begin{definition}
Let $V:$ $\mathcal{X} \rightarrow \mathcal{S}(H)$
be a classical-quantum channel. 
We denote
\[\sigma_{\mathcal{X}}:= \frac{1}{v}\sum_{x\in \mathcal{X}}  V(x)\text{ .} \]
Let $f :  \mathcal{S} \times \mathcal{X} \rightarrow \mathcal{A}$ be
the functional form of a mosaic of a $(u, m, k, \lambda_1, \lambda_2)$ GDD.
Let  $v$ and   $r$ be defined as in Section \ref{Mocd}. We denote
\begin{align*}& C(V, u, m, k, v, r, \lambda_1, \lambda_2 )\allowdisplaybreaks\notag\\
&:=\frac{r-\lambda_1}{kr} \sum_{x} \exp\left(D_2\left(V(x)\parallel  
\sigma_{\mathcal{X}}\right)\right)\allowdisplaybreaks\notag\\
&~+\frac{u(\lambda_1-\lambda_2)}{kr}\frac{1}{m} \sum_i
 \exp\left(D_2\left(\frac{1}{u}\sum_{x\in\mathcal{X}_i}V(x){\displaystyle \parallel } 
\sigma_{\mathcal{X}}\right)\right)\allowdisplaybreaks\notag\\
&~+\frac{v\lambda_2}{kr}\text{ .}
\end{align*}
\end{definition}

Assume that the confidential
messages to be transmitted are represented by 
the random variable $A$ on $\mathcal{A}$. The random
seed is represented by $S$, uniformly distributed 
on $\mathcal{S}$ and independent of $A$.
Assume the output of $S$ is $s$, and
the message $\alpha$ has to be send.


We formulate our security bound for privacy amplification
in terms of the Holevo quantity. According to \cite{Ho}
and  \cite{Sch/Wes}, the eavesdropper can never obtain more information
asymptotically than the Holevo $\chi$ quantity, no matter which
strategy  the eavesdropper uses.

\begin{theorem}\label{lfsoma}
Let $V:$ $\mathcal{X} \rightarrow \mathcal{S}(H)$
be a classical-quantum channel.

Let $f :  \mathcal{S} \times \mathcal{X} \rightarrow \mathcal{A}$ be
the functional form of a mosaic of a $(u, m, k, \lambda_1, \lambda_2)$ GDD.
We have 
\begin{align}& \max_{A\in P(\mathcal{A})}\exp\left(\frac{1}{|\mathcal{S}|}\sum_s\chi (A ;Z_s)\right)
\leq C(V, u, m, k, v, r, \lambda_1, \lambda_2 )\text{ .}
\label{maxchiPZ}\end{align}
Here 
 $Z_s$ $=\{Z_{s}(\alpha): \alpha\in \mathcal{A}\}$.   
\end{theorem}

\begin{proof}
Since the exponential function is convex, 
we have
\[\exp\left(\frac{1}{|\mathcal{S}|} 
\sum_s \chi (A ;Z_s)\right) \leq\frac{1}{|\mathcal{S}|} 
\sum_s \exp\left(\chi (A ;Z_s)\right)\text{ .}\]
We fix one $s$. 
Let $V_{s}'$ be the channel
$\mathcal{A}\rightarrow \mathcal{S}(H')$ defined by
$\alpha \rightarrow Z_{s}(\alpha)$.
       By the quantum information radius (cf. \cite{Mo}),
we have for fixed $s$:
\begin{align*}& \max_{A\in P(\mathcal{A})} \chi (A ;Z_s)\allowdisplaybreaks\notag\\
&=\max_{A\in P(\mathcal{A})}\chi\left(A ;
\{V_{s}'(\alpha): \alpha\in \mathcal{A}\}\right)\allowdisplaybreaks\notag\\
&= \min_{\sigma}\max_{\alpha\in \mathcal{A}}D\left(V_{s}'(\alpha)\parallel  \sigma\right)\allowdisplaybreaks\notag\\
&\leq \max_{\alpha\in \mathcal{A}}  D\left( V_{s}'(\alpha)\parallel  \sigma_{\mathcal{X}}\right) 
\text{ .}
\end{align*}

By
(\ref{ldarsld}) and (\ref{ldarsld2}),
it holds
\[ D\left(V_{s}'(\alpha)\parallel  \sigma_{\mathcal{X}}\right)
 \leq D_2\left( V_{s}'(\alpha)\parallel  \sigma_{\mathcal{X}}\right)\text{ .} \]

Therefore, we have
 \begin{equation}\max_{A\in P(\mathcal{A})} \exp\left(\frac{1}{|\mathcal{S}|} 
\sum_s \chi (A ;Z_s)\right)
\leq \frac{1}{|\mathcal{S}|}\max_{\alpha\in A}\sum_s
\exp\left( D_2\left( V_{s}'(\alpha)\parallel  \sigma_{\mathcal{X}}\right)\right)
\text{ .} \label{maxchiPZleq} \end{equation} 
\vspace{0.2cm}

Let $\{| v_i \rangle : i=1,\cdots,d\}$  be an arbitrary orthonormal
basis on  $H$. Let $ \langle v_j | V(x)| v_i \rangle$
$:=a_{j,i}(x)$.  We denote 
\[\rho_{\mathcal{X}} :=\left( \begin{array} {r}\bigl(a_{j,i}(x_{1})\bigr)_{j, i=1,\cdots,d}\\
\bigl(a_{j,i}(x_{2})\bigr)_{j, i=1,\cdots,d}\\
\cdots\\
\bigl(a_{j,i}(x_{|\mathcal{X}|})\bigr)_{j, i=1,\cdots,d} \end{array}\right)^*\text{ ,}\]
to be the  $d|\mathcal{X}|\times d$-matrix such that
 ${\rho_{\mathcal{X}}}_{i,k}$ 
$=a_{j,i}(x)$
if $k=(x-1)d+j$.
Let $N$ be the incidence matrix of a $(u, m, k, \lambda_1, \lambda_2)$ GDD.
Notice that for every $\alpha$:

\begin{align}& \bigl(\rho_{\mathcal{X}}(id_{H} \otimes N_\alpha) \bigr)
\bigl(\rho_{\mathcal{X}}(id_{H} \otimes N_\alpha)\bigr)^*\allowdisplaybreaks\notag\\
&= \sum_s\sum_{x: f(s,x)=\alpha} V(x)^2 \allowdisplaybreaks\notag\\
&= k^2\sum_s Z_{s}(\alpha)^2 \text{ .}\label{vsrihn} \end{align}

Now we have
\begin{align}&
\frac{1}{|\mathcal{S}|}\sum_s\exp  \left(D_2\left(Z_{s}(\alpha)\parallel \sigma_{\mathcal{X}}\right)\right)\allowdisplaybreaks\notag\\
&=\frac{1}{|\mathcal{S}|}\sum_s\mathrm{tr}\bigl(
Z_{s}(\alpha)^2\sigma_{\mathcal{X}}^{-1}\bigr)\allowdisplaybreaks\notag\\
&=\frac{1}{k^2|\mathcal{S}|}\mathrm{tr}\bigl(\rho_{\mathcal{X}}(id_{H} \otimes N)
(id_{H} \otimes N)^*\rho_{\mathcal{X}}^*\sigma_{\mathcal{X}}^{-1} \bigr)\allowdisplaybreaks\notag\\
&=\frac{1}{krv}\mathrm{tr}\bigl(\rho_{\mathcal{X}}(id_{H} \otimes N)
(id_{H} \otimes N)^*\rho_{\mathcal{X}}^*\sigma_{\mathcal{X}}^{-1} \bigr)\allowdisplaybreaks\notag\\
&=\frac{r-\lambda_1}{krv}\mathrm{tr}\biggl(\rho_{\mathcal{X}} (id_{H}\otimes id_{\mathbb{C}^{|\mathcal{X}|}})\rho_{\mathcal{X}}^*
\sigma_{\mathcal{X}}^{-1} \biggr)\allowdisplaybreaks\notag\\
&~+\frac{\lambda_1-\lambda_2}{krv}
\mathrm{tr}\biggl(\rho_{\mathcal{X}} (id_{H}\otimes C)\rho_{\mathcal{X}}^*
\sigma_{\mathcal{X}}^{-1} \biggr)\allowdisplaybreaks\notag\\
&~+\frac{\lambda_2}{krv}\mathrm{tr}\biggl(\rho_{\mathcal{X}}( id_{H}\otimes J)\rho_{\mathcal{X}}^*
\sigma_{\mathcal{X}}^{-1}\biggr)\allowdisplaybreaks\notag\\
&=\frac{r-\lambda_1}{krv}\sum_{x\in\mathcal{X}}\mathrm{tr}\biggl(V(x)^2
(\sigma_{\mathcal{X}})^{-1} \biggr)\allowdisplaybreaks\notag\\
&~+u\frac{\lambda_1-\lambda_2}{krv}\frac{1}{m} \sum_i
\mathrm{tr}\biggl((\frac{1}{u} \sum_{x\in\mathcal{X}_i} V(x))^2
\sigma_{\mathcal{X}}^{-1} \biggr)\allowdisplaybreaks\notag\\
&~+\frac{\lambda_2}{krv}v^2\sum_x \mathrm{tr}\sigma_{\mathcal{X}}\allowdisplaybreaks\notag\\
&=\frac{r-\lambda_1}{krv} \sum_{x} \exp\left( D_2\left(V(x)\parallel  
\sigma_{\mathcal{X}}\right)\right)\allowdisplaybreaks\notag\\
&~+u\frac{\lambda_1-\lambda_2}{krv} \frac{1}{m} \sum_i
 \exp\left(D_2\left(\frac{1}{u}\sum_{x\in\mathcal{X}_i}V(x){\textstyle \parallel } 
\sigma_{\mathcal{X}}\right)\right)\allowdisplaybreaks\notag\\
&~+\frac{\lambda_2}{kr}v\text{ .}
\label{1SD2Vs}\end{align} 
The first equation  is the definition
of $D_2$.
 The second equation holds because of
(\ref{vsrihn}).
 The third equation holds because of (\ref{skvr}). The fourth equation holds because of
(\ref{nnr-li}). The  fifth 
equation holds because, by the definitions of
$\rho_{\mathcal{X}}$, $C$, and $J$,
we have $\rho_{\mathcal{X}} \rho_{\mathcal{X}}^*$
$=$ $\sum_{x\in\mathcal{X}}V(x)^2$, $\rho_{\mathcal{X}} 
( id_{H}\otimes C)\rho_{\mathcal{X}}^*$
$=$ $\frac{u}{m} \sum_i
\mathrm{tr}\Bigl((\frac{1}{u} \sum_{x\in\mathcal{X}_i} V(x))^2\Bigr)$
and  $\rho_{\mathcal{X}} 
( id_{H}\otimes J)\rho_{\mathcal{X}}^*$
$=$ $v^2\sigma_{\mathcal{X}}^2$.
 The sixth equation  is again the definition
of $D_2$.

(\ref{maxchiPZ}) follows from (\ref{maxchiPZleq}) and (\ref{1SD2Vs}).
\end{proof}

\vspace{0.4cm}

\begin{remark} 
Theorem \ref{lfsoma} shows how many
confidential messages are maximal
possible when a degree of security is required  and shows how
this can be achieved  using the functional form of
a mosaic of GDDs and  BIBDs.\end{remark}

 \vspace{0.2cm}

Since a GDD with $\lambda_1=\lambda_2$ is a BIBD,
the following corollary is a consequence of 
Corollary \ref{lfsoma}:
\begin{corollary}\label{lfsoma2}
Let $f :  \mathcal{S} \otimes \mathcal{X} \rightarrow A$ be
the functional form of a mosaic of a $(v, k, \lambda)$ BIBD.

We have 
\begin{align}& \max_{P\in P(A)}\exp\left(\frac{1}{|\mathcal{S}|}\sum_s\chi (P ;Z_s)\right)\allowdisplaybreaks\notag\\
&\leq \left(1-\frac{r-\lambda}{kr}\right)
+\frac{r-\lambda}{kr}
\frac{1}{v}\sum_{x} \exp\left(D_2\left(V(x)\parallel  
\sigma_{\mathcal{X}}\right)\right)\text{ .}
\label{maxchiPZ2}\end{align}
\end{corollary}\vspace{0.2cm}


We denote the quantum states
at the output of $V$ by $\eta^{Z}$
We can describe classical-quantum hybrid system 
$ASZ$ by an ensemble $\eta^{ASZ}$
when we embody the classical $A$ and $S$  into a  $\lvert \mathcal{A} \rvert$ dimensional
 Hilbert space  $H_{\mathcal{A}}$.
   with orthonormal basis
  $\{| \alpha \rangle : \alpha\in \mathcal{A}\}$ and
  a  $\lvert \mathcal{S} \rvert$ dimensional
  space  $H_{\mathcal{S}}$ with orthonormal basis
  $\{| s \rangle : s\in \mathcal{S}\}$, respectively.
The quantum state of the
system $ASZ$ is
\[\eta^{ASZ}:=\frac{1}{|\mathcal{A}|} \frac{1}{|\mathcal{S}|}\sum_{\alpha\in \mathcal{A}}\sum_{s\in \mathcal{S}} |\alpha \rangle \langle \alpha|\otimes
|s \rangle \langle s| \otimes  Z_s(\alpha) \text{ .}\]
By Theorem \ref{lfsoma}, we can bound $\left\|\eta^{ASZ}-\eta^{AS}\otimes \eta^{Z}\right\|_1$
for all possible message distributions
by the following corollary:

\begin{corollary}
\label{wmsv-}
Let $f :  \mathcal{S} \otimes \mathcal{X} \rightarrow \mathcal{A}$ be
the functional form of a mosaic of a $(u, m, k, \lambda_1, \lambda_2)$ GDD.
By Theorem \ref{lfsoma} we have 
\begin{align}&  
\left\|\eta^{ASZ}-\eta^{AS}\otimes \eta^{Z}\right\|_1\allowdisplaybreaks\notag\\
&\leq \sqrt{ 2\ln 2\log
C(V, u, m, k, v, r, \lambda_1, \lambda_2 )}
\text{ .}
\end{align}

\end{corollary}
\begin{proof}
We have 
\begin{align*}&\eta^{Z}\allowdisplaybreaks\\
&=\mathrm{tr}_{AS}(\eta^{ASZ})\allowdisplaybreaks\\
&=\frac{1}{|\mathcal{A}|} \frac{1}{|\mathcal{S}|}
\sum_{\alpha\in \mathcal{A}}\sum_{s\in \mathcal{S}}  Z_s(\alpha)\allowdisplaybreaks\\
&=\sigma_{\mathcal{X}}
\text{ .}\end{align*}

By the triangle inequality, it holds
\begin{align*}&\left\|\eta^{ASZ}-\eta^{AS}\otimes \eta^{Z}\right\|_1\allowdisplaybreaks\\
&=\left\|\frac{1}{|\mathcal{A}|} \frac{1}{|\mathcal{S}|}\sum_{\alpha\in \mathcal{A}}
 \sum_{s\in \mathcal{S}}  Z_s(\alpha)-\sigma_{\mathcal{X}}\right\|_1\allowdisplaybreaks\\
&\leq\frac{1}{|\mathcal{A}|} \frac{1}{|\mathcal{S}|}\sum_{\alpha\in \mathcal{A}}
 \sum_{s\in \mathcal{S}} \left\| Z_s(\alpha)-\sigma_{\mathcal{X}}\right\|_1\allowdisplaybreaks\\
&\leq\max_{\alpha\in \mathcal{A}}\frac{1}{|\mathcal{S}|}
\sum_{s\in \mathcal{S}} \left\|Z_s(\alpha)-\sigma_{\mathcal{X}}\right\|_1
\text{ .}\end{align*}
Since for nonnegative $\{a_{1}, \cdots, a_{N}\}$ we have
$(\frac{1}{N} \sum_{i=1}^{N} a_{i})^2\leq \frac{1}{N} \sum_{i=1}^{N} a_{i}^2$, it holds
\[\left\|\eta^{ASZ}-\eta^{AS}\otimes \eta^{Z}\right\|_1^2
\leq \max_{\alpha}\frac{1}{|\mathcal{S}|}  \sum_{s\in \mathcal{S}} \left\|Z_s(\alpha)-
\sigma_{\mathcal{X}}\right\|_1^2\text{ .}\] By the Quantum Pinsker inequality (cf. \cite{Hay}) we have
\[\|Z_s(\alpha)-\sigma_{\mathcal{X}}\|_1^2  
\leq 2\ln 2 D\left(V_{s}'(\alpha) \parallel \sigma_{\mathcal{X}} \right)\text{ .}\] 
Since the exponential function is convex, 
we have
\begin{align*}&\exp\Bigl(\frac{1}{2\ln 2} \left\|\eta^{ASZ}-\eta^{S}\otimes \eta^{Z}\right\|_1^2\Bigr) \allowdisplaybreaks\\
&\leq \max_{\alpha}\exp\left( 
 \frac{1}{|\mathcal{S}|}\sum_{s\in \mathcal{S}} D\bigl(V_{s}'(\alpha) \parallel \sigma_{\mathcal{X}} \bigr)\right)\allowdisplaybreaks\\
 &\leq \max_{\alpha}\frac{1}{|\mathcal{S}|} \sum_{s\in \mathcal{S}}
\exp\left(D\left(V_{s}'(\alpha) \parallel \sigma_{\mathcal{X}}\right)\right)\text{ .}\end{align*}

Corollary \ref{lfsoma} follows from (\ref{1SD2Vs}).
\end{proof}
This means that we can archive that the eavesdropper's observations are  nearly independent
of the message.




Theorem \ref{lfsoma} 
shows how much randomness  is sufficient in the randomized inverse to obtain a given
level of semantic security. We apply the derandomization
technique to  construct a semantic-security code without common randomness using
a transmission code and a common-randomness semantic-security code with
appropriate error scaling.

\subsection{Secrecy Rate}\label{Groaom}

Suppose we have a $(n, J_n)$ public  code with small decoding error.
    In order to achieve semantic security in the presence 
of an eavesdropper's 
channel $V$, the functional form $f:\mathcal X\times\mathcal S\to\mathcal A$ of a mosaic of $(J_n,k,\lambda)$ BIBDs is applied as the security function. By Theorem \ref{lfsoma}, 
the information leakage $\max_A\frac{1}{|\mathcal S|}\sum_s\chi(A;Z_s)$, 
using $\log(1+t)\leq t$, can be upper-bounded by
$
    \frac{r-\lambda}{kr}\sum_x\exp(D_2(V(x)\Vert\sigma_{\mathcal X}))$.

In order for this to be small, $k$ should be sufficiently larger than $\sum_x\exp(D_2(V(x)\Vert\sigma_{\mathcal X}))$. 
For given positive 
$\epsilon_{leak}$
 for a small positive $\lambda$ when $n$ is sufficiently large then for
\[k=\sum_{x^n}\exp(D_2(V^n(x^n)\Vert\sigma_{\mathcal X}^n))+\lambda\]
the information leak is bounded by $\epsilon_{leak}$.
This means that the number of messages
     per $n$ channel uses 
in a semantically secure way reduces from $J_n$ to $\frac{J_n}{k}$.      In information theory, it is more common to talk about rates by taking the logarithm of the number of messages. Then the rate reduces from 
    $\frac{\log J_n}{n}$ to 
\begin{align*}&\frac{\log|{\mathcal A}_n|}{n}\\
&=\frac{\log \frac{J_n}{k}}{n}\\
&=\frac{\log \frac{J_n}{k}-\sum_x\exp(D_2(V(x)\Vert\sigma_{\mathcal X}))-\lambda}{n}\text={ .}\end{align*}
 The rate loss
$
    R_f=\frac{\log J_n-\log|{\mathcal A}_n|}{n}
$
can be regarded as the rate of the security function. The rate of the functional form of a mosaic of GDDs is defined in the same way, but the relation to security is more complicated since it depends on the type of the underlying GDDs (see \cite{Bo/Wi} for a detailed discussion).
By \cite{Ho}, \cite{Sch/Ni},
and  \cite{Sch/Wes}, there exist reliable public code
with rate $J_n=\max_A \chi (A, W(A))$. Thus, 
   with ous approach,
   we can archive
the semantic secrecy rate
\[\max_A \chi (A, W(A))-D_2(V(A)\Vert\sigma_{\mathcal X})\text{ .}\]

For our modular code, we need common randomness
as an additional resource.
As \cite{Bo/No} showed, the common randomness
is a very ``costly'' resource.
 Thus, our next step is a process known as
derandomization, when 
the sender generates
the seeds and send it
to the receiver with the public code. 
The standard derandomization works as follows:
Every code word $\mathcal{C}^{det} $ $=$ $(E^{\mu(|S|)+n},\{G_{\alpha}^{\mu(|S|)+n} :
\alpha\})$ is 
a composition of words of a public code
 $(E^{\mu(n)},\{G_{s}^{\mu(|S|)} :
s\})$, and words of a semantic-secure modular 
code  $\{(E^{n}_{s},\{G_{s,\alpha}^{n} :
\alpha\}):s\}$.
Here  $\mu(|S|)$ is the length of 
the first code words.  The sender  chooses a $s$ uniformly at random from $\mathcal S$
and uses the public code  $(E^{\mu(|S|)},\{G_{s}^{\mu(|S|)} :
s\})$
to send it to the receiver as a pre-code (recall
$s$ do not have to be secure).
Thus we have a code with encoder
$E^{\mu(|S|)+n}((x^{\mu(|S|)}x^{n})|\alpha)$ 
$= \frac{1}{|S|}\sum_{s}E^{\mu(|S|)}(x^{\mu(|S|)}|s)
E^{n}_{s}(x^{n}|\alpha)$ and decoder operators
$G^{\mu(n)+n}_{s,\alpha}$ 
$= \sum_{s} G_{s}^{\mu(|S|)}\otimes G_{s,\alpha}^{n}$,
      which consists of
alternate public  code words and semantic-secure modular 
code words, where we use the public  code to generate the seed, and use it only once in
the semantic-secure modular code. 
   However, when the
size of the seed set is too large, then
$\mu(|S|)$ may be also too large, and
this standard derandomization technique
may cause significant rate loss. 
      Follow the idea
of \cite{Wi/Bo} and \cite{Bo/Ca/De/Fe/Wi},
we reduce the total size of channel uses by
reusing one seed for multiple semantic-secure modular code words.     
Instead of one  two-part code word for
every single message, 
we build a $(N+1)$-tuple of codewords. 
Each tuple is a composition of a public 
codeword that generates the seed and $N$ 
semantic secure modular code words to 
transmit $N$ messages to the intended receiver.
The code $\mathcal{C}^{det} $ $=$ $(E^{\mu(|S|)+nN},\{G_{\alpha}^{\mu(|S|)+nN} :
\alpha\})$ 
consists of encoder
$E^{\mu(|S|)+nN}((x^{\mu(|S|)}x^{nN})|\alpha)$ 
$= \frac{1}{|S|}\sum_{s}E^{\mu(|S|)}(x^{\mu(|S|)}|s)
E^{n}_{s}(x^{n}|\alpha_{1})\cdots
E^{n}_{s}(x^{n}|\alpha_{N})$ and decoder operators
$G^{\mu(n)+nN}_{\alpha}$ 
$= \sum_{s} G_{s}^{\mu(|S|)}\otimes \bigotimes_{j=1}^{N} G_{s,\alpha_{j}}^{n}$.
The 
$N$ 
semantic secure modular code words in every
tuple share one single seed.
 We choose a squence
$N(n)$ such that $1\ll N(n)$,
$N(n)\ll P_e(\mathcal{C}, n)^{-1}$
(or $N(n)\ll  P_e^{m}(\mathcal{C}, n)^{-1}$),
and $N(n)\ll \epsilon_{leak}^{-1}$
    With this approach,
the semantic secrecy rate above
 can be achieved.

\section{Further Remark}

Several efficiently computable examples of mosaics of BIBDs and of GDDs are given in \cite{Bo/Wi}. Efficient computability of a mosaic here means that the functional form as well as the random choice of an $x\in\mathcal X$ given a seed $s$ and a message $\alpha$ can be computed in polynomial time. All these examples are constructed in such a way that the seed set $\mathcal S$ should be as small as possible compared with $\mathcal X$ for the given rate $R_f$. It turns out that $\log|\mathcal S|\geq2R_f\log|\mathcal X|$ roughly for $R_f\geq \frac{1}{2}$ and that this lower bound is tight, both for mosaics of BIBDs and of GDDs. If $R_f<\frac{1}{2}$, then in mosaics of BIBDs, the best one can hope for is that $|\mathcal S|$ is at least as large as $|\mathcal X|$, and this can be approximated for sufficiently large $\mathcal X$. For mosaics of GDDs, one can still achieve $\log|\mathcal S|=2R_f\log|\mathcal X|$, but only at the price that the rate loss in order to achieve a given security level from Theorem \ref{lfsoma} is suboptimally large except for special channels adapted to the functional form of the mosaic. An optimal rate loss in the case of mosaics of GDDs with $R_f<\frac{1}{2}$ still requires $|\mathcal S|\geq|\mathcal X|$. Details for classical wiretap channels can again be found in \cite{Bo/Wi}, and via Theorem \ref{lfsoma} they carry over to the classical-quantum wiretap channels treated here.

\begin{ex}\label{ex2}

We would like to mention here another security function which has appeared before in the literature for which it has been proved that it achieves the secrecy capacity of discrete memoryless wiretap channels with semantic security, but which has so far not been recognized as the functional form of a mosaic of designs. This function is the example of a different concept of security functions underlying the analysis in \cite{Hay/Mat} Remark 16. It has also been discussed in some detail in \cite{Wi/Bo}, Appendix C. 

Set $\mathcal X=\mathbb F_{q^t}$ and let $\ell$ be a positive integer smaller than $t$. As message set, take any $(t-\ell)$-dimensional subspace $\mathcal A$ of $\mathbb F_q^t$ and also choose any $\ell$-dimensional subspace $\mathcal V$ of $\mathbb F_q^t$ satisfying $\dim(\mathcal A\cap\mathcal V)=0$. Then set $\mathcal S_1=\mathbb F_{q^t}^*$ and $\mathcal S=\mathcal S_1\times\mathcal A$ and define $f:\mathcal S\times\mathcal X\to\mathcal A$ by 
\[
    f(s,x)=\alpha\qquad\text{if }s_1x+s_2\in \alpha+\mathcal V,
\]
where $s=(s_1,s_2)$. Given a basis of $\mathcal V$, the computation of this function and its randomized inverse can be done efficiently.
We claim that $f$ is the functional form of a mosaic of BIBDs.

In order to prove this claim, we need to show that the sets 
\[
    \{x:s_1x+s_2=\alpha+\mathcal V\}=s_1^{-1}(\alpha-s_2+\mathcal V).
\]
for each $s=(s_1,s_2)\in\mathcal S$ and every $\alpha\in\mathcal A$ are the blocks of a BIBD on $\mathcal X$. First of all, we note that all of these sets are of size $k=q^\ell$. Now choose any $x\neq x'\in\mathcal X$. Then
\begin{align*}&
    \{s:f(s,x)=f(s,x')=\alpha\}\\
    &=\{(s_1,s_2):s_1(x-x')\in\mathcal V,s_1x+s_2\in \alpha+\mathcal V\}.
\text{ .}\end{align*}
The size of this set equals
\[
    \lambda
    =\sum_{s_1\in(x-x')^{-1}\mathcal V\setminus\{0\}}\lvert\mathcal A\cap(\alpha-s_1x+\mathcal V)\rvert
    =q^\ell-1.
\]
Thus the sets    $\{x:f(s,x)=\alpha\}$ for $s\in\mathcal{S}$     form the blocks of a BIBD on $\mathcal X$ with parameters
\[
    v=q^t,\quad b=q^{t-\ell}(q^t-1),\quad r=q^t-1,\quad k=q^\ell,\quad \lambda=q^\ell-1.
\]
The rate $R_f=\frac{t-\ell}{t}$ can be chosen flexibly. Note that $\log|\mathcal S|\approx(1+R_f)\log|\mathcal X|>2R_f\log|\mathcal X|$ if $R_f<1$, so $f$ is not optimal in terms of seed length as discussed in \cite{Bo/Wi}.

\end{ex}

\section*{Acknowledgment}The work of H. Boche 
supported by the
German Federal Ministry of Education and Research (BMBF) under Grants
16KIS0858 and 16KIS0948, and by the German Research
Foundation (DFG) within the Gottfried Wilhelm Leibniz Prize under Grant BO
1734/20-1 and within Germany's Excellence Strategy EXC-2111 - 390814868
and EXC-2111 - 390814868. The work of M. Wiese was  supported  by the German Research
Foundation (DFG) within the  Germany's Excellence Strategy - EXC 2092 CASA-390781972.
The work of M. Cai was  supported  by the German Research
Foundation (DFG) within the Walter Benjamin-Fellowship CA 2779/1-1

\end{document}